\documentclass{article}

  \usepackage[preprint]{neurips_2025}

\usepackage[utf8]{inputenc} 
\usepackage[T1]{fontenc}    
\usepackage{hyperref}       
\usepackage{url}            
\usepackage{booktabs}       
\usepackage{amsfonts}       
\usepackage{nicefrac}       
\usepackage{microtype}      
\usepackage{xcolor}         
\usepackage{amsmath, amssymb}
\usepackage{graphicx}
\usepackage{mathtools}
\usepackage{float}

\title{The LLM as a Network Operator: A Vision for Generative AI in the 6G Radio Access Network}

\author{%
  Oluwaseyi~Giwa\thanks{corresponding author: oluwaseyi@aims.ac.za} \\
  African Institute for Mathematical Sciences\\
  Muizenberg, 7945 \\
  \texttt{oluwaseyi@aims.ac.za} \\
   \AND
   Michael~Adewole\\
   Olabisi Onabanjo University\\
   Ago-Iwoye, Nigeria\\
   \texttt{michaseyi@gmail.com}\\
   \AND
   Tobi~Awodumila \\
   African Institute for Mathematical Sciences \\
  Muizenberg, 7945 \\
  \texttt{tobi@aims.ac.za} \\
  \AND
  Pelumi Aderinto \\
  African Institute for Mathematical Sciences \\
  Muizenberg, 7945 \\
  \texttt{pelumi@aims.ac.za} \\
}

\newtheorem{theorem}{Theorem}[section]

\newtheorem{lemma}[theorem]{Lemma}
\newtheorem{definition}{Definition}[section]
\newtheorem{prop}[theorem]{Proposition}
\newtheorem{proof}{Proof}[section]

\begin{document}

\maketitle

\begin{abstract}
  The management of future AI-native Next-Generation (NextG) Radio Access Networks (RANs), including 6G and beyond, presents a challenge of immense complexity that exceeds the capabilities of traditional automation. In response, we introduce the concept of the LLM-RAN Operator. In this paradigm, a Large Language Model (LLM) is embedded into the RAN control loop to translate high-level human intents into optimal network actions. Unlike prior empirical studies, we present a formal framework for an LLM-RAN operator that builds on earlier work by making guarantees checkable through an adapter aligned with the Open RAN (O-RAN) standard, separating strategic LLM-driven guidance in the Non-Real-Time (RT) RAN intelligent controller (RIC) from reactive execution in the Near-RT RIC, including a proposition on policy expressiveness and a theorem on convergence to stable fixed points. By framing the problem with mathematical rigor, our work provides the analytical tools to reason about the feasibility and stability of AI-native RAN control. It identifies critical research challenges in safety, real-time performance, and physical-world grounding. This paper aims to bridge the gap between AI theory and wireless systems engineering in the NextG era, aligning with the AI4NextG vision to develop knowledgeable, intent-driven wireless networks that integrate generative AI into the heart of the RAN.
\end{abstract}

\section{Introduction}\label{intro}
Next-Generation (NextG) Radio Access Networks (RANs), including 6G and future WiFi standards, are growing increasingly complex, driven by network densification, dynamic spectrum sharing, and the need to support heterogeneous services with conflicting quality of service (QoS) requirements (e.g., eMBB, URLLC). This complexity motivates new automation paradigms. In the context of AI4NextG research, this complexity is not just a challenge but an opportunity: wireless environments are dynamic, high-dimensional, and partially observable\textemdash characteristics that make them ideal testbeds for developing robust, adaptive, and explainable AI/ML solutions. 

Open RAN (O-RAN) aims to ``bring openness and intelligence'' into traditionally closed RANs \citep{azariah2024}. In particular, intent-driven management is emerging as a key abstraction: e.g., RFC 9315 defines \textit{intent} as a high-level, declarative specification of network goals, a notion adopted by 3GPP and TM Forum \citep{tmforum2023, xu2023, bimo2025}. At the same time, transformer-based large language models (LLMs) have shown a remarkable ability to interpret unstructured commands and generate plans. In line with the AI4NextG vision of AI-native protocol and architecture design, we propose LLM-RAN operators: architectures in which an LLM is embedded into the RAN control plane to translate user intents (in natural language or formal policy terms) into network configurations.

Our approach is distinct from existing work because we provide mathematical definitions, analytical propositions, and convergence guarantees that allow rigorous reasoning about feasibility, stability, and expressiveness, properties rarely addressed in prior LLM-for-RAN research. This analytical treatment not only grounds the vision of generative AI in the RAN but also creates a foundation for formal verification, making it relevant for both research and NextG standardization.

Beyond intent translation, the proposed framework directly supports core AI4NextG problem domains such as dynamic spectrum access, where LLM-guided policies can coordinate frequency and power allocation under regulatory and interference constraints; cross-layer machine learning optimizations for joint sensing, control, and communication; and edge AI deployment for low-latency decision-making in distributed RAN architectures. While pioneering works such as RANGPT, WiLLM, LLM-xApp, etc., have demonstrated proof-of-concept integrations, \citep{polese2022, rangpt2023, liu2025, xingqi2025, gajjar2025}, a formal framework for analyzing the stability, expressiveness, and convergence of such AI-native RAN control systems, critical for NextG standardization and deployment, is still missing. We therefore introduce such a framework.

Contemporary RAN surveys note that each generation (D-RAN, C-RAN, vRAN, Open RAN) seeks more intelligence and flexibility \citep{azariah2024, mathushaharan2025}. In particular, O-RAN architecture defines two key controllers: a near-real-time (RT) RAN intelligent controller (RIC) and a non-RT RIC \citep{polese2022}. LLMs naturally fit as cognitive components in these RICs. For example, we can view the non-RT RIC as hosting an LLM-based ``Strategy Agent'' that interprets intent, while the near-RT RIC hosts reactive ``Operator Agent'' xApps. Aira's RANGPT \citep{rangpt2023} already demonstrates conversational LLM interfaces to RAN data, yielding energy-saving and performance-optimization commands.

\section{Background and Related work}\label{related-work}
LLM-driven network control lies at the intersection of intent-based networking, O-RAN architectures, and agentic AI. In O-RAN, the Non-RT RIC handles strategic policy and long-term optimization (> 1s) while the Near-RT RIC manages radio resource control loops (10 ms-1s) via xApps. This modular separation enables the integration of AI-driven agents, including LLM-based rApps for policy translation and coordination \citep{polese2022, azariah2024}.

Recent work has explored LLMs in this context:
\begin{itemize}
    \item RANGPT \citep{rangpt2023}: Conversational LLM interface to RAN data for energy and performance optimization.
    \item WiLLM \citep{liu2025}: LLM inference for wireless slicing with novel ``Tree-Branch-Fruit'' architecture.
    \item LLM-xApp \citep{xingqi2025}: Meta-prompt-driven resource allocation refinement for QoS targets.
    \item ORANSight-2.0 \citep{gajjar2025}: Open-source, fine-tuned foundational LLMs for O-RAN tasks.
\end{itemize}

While these works demonstrate empirical feasibility, none provide a formal framework to analyze expressiveness, stability, or convergence. By combining AI theory with RAN-specific constraints, we enable principled evaluation of AI-native RAN control loops.

\section{The LLM-RAN Operator Concept}\label{concept}
We define an LLM-RAN operator as an abstract mapping that takes a user intent and network state, and outputs RAN control actions. Formally, let \(\mathcal{S}\) be the state space of the RAN (e.g., measurements and configurations of all cells), \(\mathcal{I}\) the intent space (natural-language or formal commands), and \(\mathcal{A}\) the action space (feasible reconfigurations, e.g., power levels, slice settings). We model actions as well-typed commands in a finite \emph{RAN Command DSL} $\mathcal{L}$, and require every emission to pass a total allow-list \emph{validator} \(V:\mathcal{L}\!\to\!\{0,1\}\) that enforces units, bounds, and safety before compilation to O-RAN interfaces (A1/E2/O1).

\begin{definition}[Intent]
An intent \(i \in \mathcal{I}\) is a high-level specification of goals or constraints (e.g. ``maximize cell-edge throughput subject to power budget''), as defined by standards \citep{clemm2022, bimo2025}
\end{definition}

\begin{definition}[LLM-RAN Operator]
    An LLM-RAN operator is a function
    \begin{equation}
        \mathcal{O}_{\text{LLM}}: \mathcal{I} \times \mathcal{S} \rightarrow \mathcal{A} \times \mathcal{S},
    \end{equation}
\end{definition}

This operator maps an intent \(i\) and current state \(s\) to a specific control action \(a\). For example, the fine-tuned models from the ORANSight-2.0 suite can be viewed as specific, empirically realized instances of such an operator \citep{gajjar2025}. The network then transitions to a new state \(s'\) according to the environment's dynamics. In practice, the LLM may output a sequence of actions (a policy) or RAN commands to be enacted by an underlying controller.

Existing literature indicates that LLMs can serve as highly flexible translators between intents and detailed plans. For example, \citep{xingqi2025} demonstrated an LLM-based policy that iteratively refines slice allocations to maximize QoS. We capture this capability via a proposition on expressive power:

\begin{prop}[Expressiveness of LLM-Policies]
An LLM-RAN operator with sufficient model capacity can approximate any effective mapping from intents to RAN actions, assuming enough prompt context. In other words, the space of policies representable by a large transformer is universal for bounded RAN control tasks \citep{ibm2025, zeinab2025}.
\end{prop}

\begin{lemma}[Monotonic Improvement Under LLM Guidance]
Let \(U(s)\) be a utility function representing a network performance metric. Assume the LLM operator is designed to solve the single-step optimization problem \(a_t = \arg \max_{a \in \mathcal{A}}U(f_{\text{env}}(s_t, a))\). If a solution \(a_t\) exists such that \(U(s_{t + 1}) \geq U(s_t)\), then the sequence of states generated by the system is monotonically improving in utility. This property is observed in LLM-guided cases (e.g., decreasing transmit power raises efficiency).
\end{lemma}


\begin{theorem}[Convergence to Fixed Point]
If the LLM-RAN operator and environment form a contraction mapping on \(\mathcal{S}\) (e.g. under a suitable norm), then the closed-loop sequence \((s_t)\) converges to a unique fixed point \(s\text{\textasciicircum}\) satisfying \(\mathcal{O}_{\rm LLM}(i,s\text{\textasciicircum}) = a\text{\textasciicircum},; f_{\rm env}(s\text{\textasciicircum}, a\text{\textasciicircum}) = s\text{\textasciicircum}\). In practice, if the LLM’s updates become smaller over time (e.g., guided by decaying exploration
), the system reaches equilibrium.
\end{theorem}

\begin{proof}[Idea]
    By Banach’s fixed-point theorem, any contraction \(s_{t+1} = F(s_t)\) converges to a unique \(s^*\). Here \(F(s) \coloneqq f_{\rm env}(s,\mathcal{O}_{\rm LLM}(i,s))\). If \(|\partial F/\partial s|<1\), repeated application yields convergence. While we do not prove contraction for a specific network, many network control policies behave stably when near an optimum. Empirical RAN studies (with or without LLM) often show rapid convergence in practice.
\end{proof}

This formalism lets us reason about feasibility: e.g., if the intent \(i\) is achievable, there exists \(s\text{\textasciicircum}\) such that \(\mathcal{O}_{\rm LLM}(i,s\text{\textasciicircum})\) yields no further changes (zero residual). We note that proving \(\text{F}\) is a contraction for a general, high-dimensional RAN environment is a major open research challenge. However, this framework provides the analytical tools to identify specific conditions on the environment \(f_{\text{env}}\) and the operator (\(\mathcal{O}_{\text{LLM}}\)) under which convergence can be formally guaranteed\footnote{Please check Appendix for all detailed proofs.}.

\section{Architectural Framework for an LLM-RAN Operator}\label{framework}
To realize the theoretical (\(\mathcal{O}_{\text{LLM}}\)) operator, we propose a modular architectural framework grounded in the O-RAN paradigm. The core principle is to decouple strategic, slow-loop reasoning from reactive, fast-loop execution. This is achieved by logically separating the components responsible for understanding intent, perceiving state, reasoning, and acting. Fig.~\ref{llm-diagram} illustrates this conceptual design.
\begin{figure}
    \centering
    \includegraphics[width=1.0\linewidth]{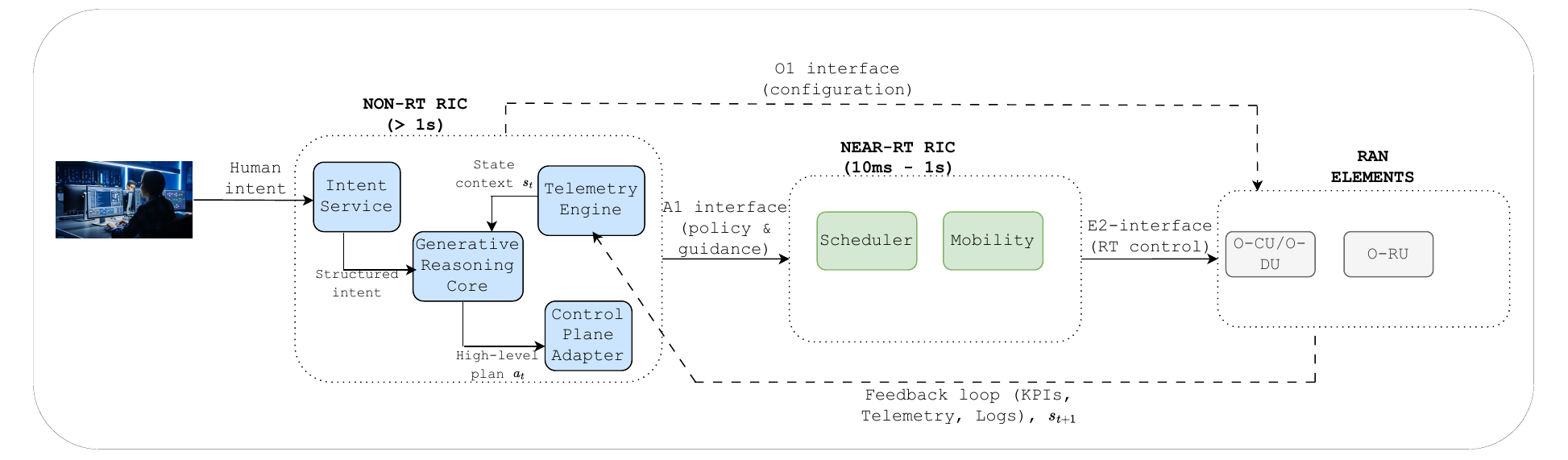}
    \caption{The architectural framework for the LLM-RAN Operator. The Intent Service and Reasoning Core typically reside in the Non-RT RIC, generating policies for the Control Adapter, which configures xApps in the Near-RT RIC.}
    \label{llm-diagram}
\end{figure}

\subsection{Core Components}
Our proposed architecture consists of 4 key logical components:
\begin{itemize}
    \item[1.] \textbf{Intent Service:} This component serves as the primary interface for human operators. Its function is to process high-level intents \(i \in \mathcal{I}\).  It must parse natural language commands, validate them against a set of permissible goals, and translate them into a structured, machine-readable format (e.g., a JSON object with objectives and constraints). This structured intent is then passed to the reasoning core.
    \item[2.] \textbf{Telemetry Engine:} This is the perception system of the operator, responsible for constructing the state vector \(s \in \mathcal{S}\). It aggregates and filters high-dimensional data from various network sources, including Key Performance Indicators (KPIs) from the O-RAN Non-RT RIC and near-real-time measurements from the near-RT RIC. A critical function of this engine is to create a coherent tokenized summary of the network's current condition or ``state context'' that is suitable for the LLM.
    \item[3.] \textbf{Generative Reasoning Core:} This is the heart of the \(\mathcal{O}_{\text{LLM}}\) operator, embodied by a fine-tuned foundation model. It receives the structured intent from the Intent Service and the tokenized summary from the Telemetry Engine. Its task is to synthesize this information and generate a high-level plan or a specific action \(a \in \mathcal{A}\). This process mirrors the ``Reason-Act'' paradigm seen in agentic AI frameworks \citep{yao2023}. The output could be an abstract directive (e.g., ``Prioritize URLLC traffic in cell'') or a set of concrete parameter changes.
    \item[4.] \textbf{Control Plane Adapter \& Execution Engine:} This component acts as the bridge between the LLM's generative output and the network's control interfaces. It performs two crucial roles, which are \textit{translation} and \textit{safety \& guardrails}. First, it translates the abstract plan from the Reasoning Core into low-level, executable commands compliant with O-RAN interfaces (e.g., A1 policies for the Near-RT RIC or O1 configurations for network elements). Additionally, it serves as a critical safety filter, ensuring that the generated actions are syntactically valid and semantically safe, preventing the LLM from issuing commands that could destabilize the network. This component effectively gives the LLM the ability to use network ``tools'' \citep{schick2023}.
\end{itemize}

\subsection{Operational Workflow and O-RAN Mapping}
The components work in a closed loop, naturally mapping onto the O-RAN architecture's separation of timescales. A human operator provides an intent (e.g., ``Reduce energy consumption in the downtown sector between midnight and 6 AM''). The \textbf{Intent Service} and \textbf{Generative Reasoning Core}, operating as an rApp within the Non-RT RIC, process this. The Reasoning Core, using data from the \textbf{Telemetry Engine}, devises a new energy-saving policy. The \textbf{Control Plane Adapter} translates this policy into A1 interface directives. For example, it might instruct the Near-RT RIC to favor a specific energy-efficient scheduling xApp. The xApps in the Near-RT RIC execute the policy, making millisecond-level decisions (e.g., adjusting scheduling weights, putting component carriers to sleep) that align with the LLM's strategic guidance. The Telemetry Engine observes the impact of these actions on the network KPIs (energy saved, user throughput), creating an updated state context \(s_{t + 1}\) and closing the loop.

This hierarchical design allows the LLM to provide high-level cognitive guidance without being burdened by the stringent latency requirements of real-time radio resource management, which remains the domain of specialized xApps.

\section{Challenges and Research Directions}\label{challenges}
While the LLM-RAN operator concept offers a transformative vision, its practical realization hinges on addressing several fundamental research challenges. Our formal framework helps to define these hurdles precisely. The most significant challenge is the inherent latency of large models. The inference time for contemporary LLMs can be hundreds of milliseconds to seconds, which is incompatible with the stringent 10ms-1s control loops of the Near-RT RIC. Our hierarchical architecture mitigates this by placing the LLM in the Non-RT RIC. Still, future work must explore techniques like model distillation, quantization, and specialized hardware accelerators to push generative intelligence closer to the real-time domain.

Furthermore, an LLM that ``hallucinates'' a network command could have catastrophic consequences, from service outages to equipment damage. While systems like ORANSight-2.0 \citep{gajjar2025} are fine-tuned to be domain-specific, they are not immune to such failures. This underscores the need for a formal approach to safety and verification, which our framework aims to enable. This necessitates research in two areas: (i) formal verification, and (ii) Explainability. 

Another key challenge is that LLMs are trained on vast corpora of text, not on the laws of physics. A key open question is how to ``ground'' the model's understanding in the complex dynamics of the wireless environment. The operator must learn that specific actions have physical consequences (e.g., increasing power also increases interference). This requires bridging the gap between the symbolic reasoning of LLMs and the continuous, high-dimensional state space (\(\mathcal{S}\)) of the RAN, likely through sophisticated multi-modal models trained in high-fidelity digital twins. Finally,  fine-tuning an LLM for the RAN domain requires vast amounts of high-quality data, including network states, corresponding actions, and resulting outcomes. Real-world data, especially for failure scenarios, is scarce, proprietary, and often noisy. Future research must focus on developing advanced simulation techniques and data-efficient learning methods (e.g., few-shot learning, reinforcement learning from human feedback) tailored to the networking domain.

\section{Conclusion}\label{conclusion}
We introduced a formal framework for the LLM-RAN Operator, a novel paradigm for intelligent and intent-driven RAN control in the NextG era. By defining the operator as a mathematical object and analyzing its theoretical properties through propositions and theorems, we move beyond empirical demonstrations to provide analytical tools for reasoning about expressiveness, stability, and convergence. The proposed O-RAN-grounded architecture separates high-level strategic reasoning from near-RT execution, enabling generative AI integration without violating stringent latency constraints.

In practice, this framework could be instantiated for high-priority AI4NextG use cases, including learning-driven spectrum allocation policies for interference-limited networks, cross-layer reinforcement learning for QoS guarantees, and generative AI agents running on edge nodes to provide localized RAN control. Addressing the significant challenges of real-time performance, safety, and physical-world grounding will require interdisciplinary collaboration between the wireless and AI communities.

We believe that positioning LLMs as cognitive agents at the heart of the RAN represents a fundamental shift from traditional automation to autonomous, intent-driven management. This is not merely an academic exercise but a necessary step to unlock the full potential of 6G and beyond, enabling networks that are not just connected, but truly intelligent.

\bibliographystyle{unsrt}

\bibliography{references}

\small


\appendix

\section{Appendices and Supplementary Material}
\subsection{Background and Related Work}
Fig.~\ref{network_evolution} illustrates the motivation: as wireless systems evolve from 1G to 6G, AI has progressively augmented control functions \citep{mathushaharan2025}.

\begin{figure}[H]
    \centering
    \includegraphics[width=1.0\linewidth]{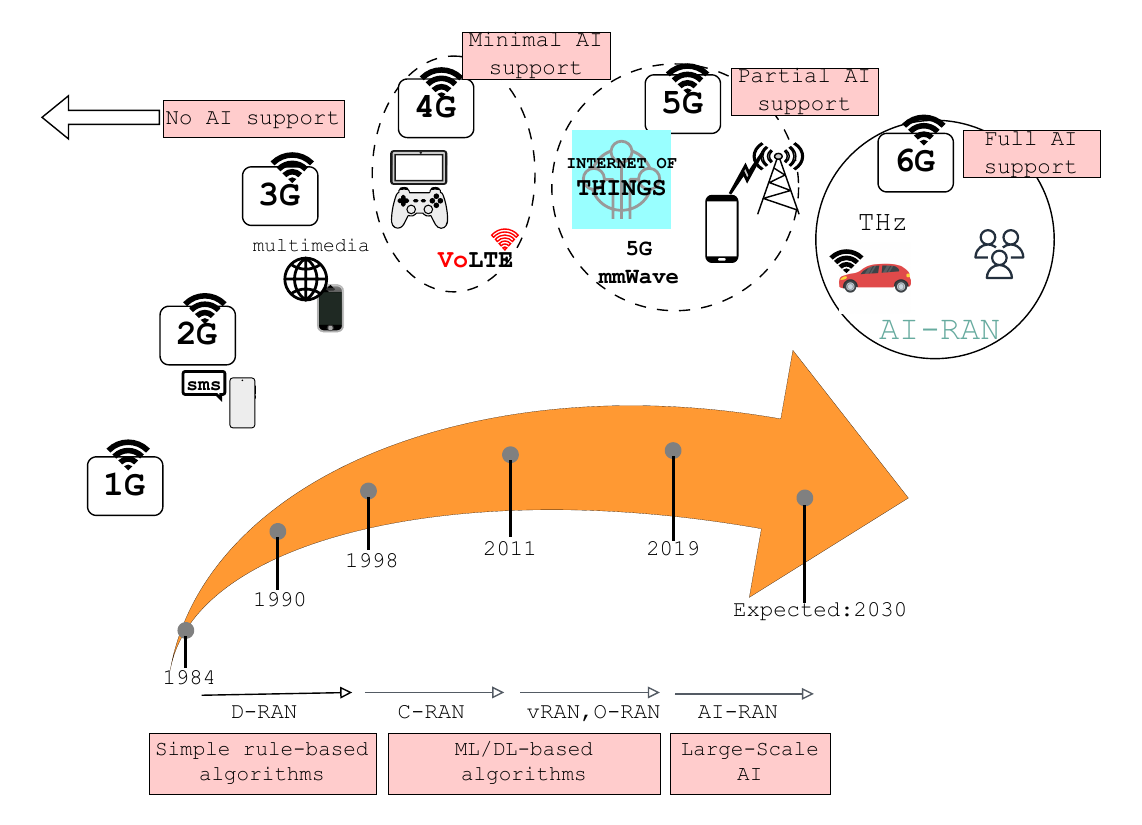}
    \caption{Evolution of mobile networks and RAN architectures with increasing AI integration}
    \label{network_evolution}
\end{figure}
\citep{clemm2022} introduced an Intent-Based RAN framework where LLMs translate JSON-encoded intents into O-RAN configurations. In particular, the LLM-hRIC framework embeds an LLM as an rApp in the non-RT RIC to provide strategic guidance for near-RT xApps \citep{lingyan2025}. \citep{xingqi2025} also presented LLM-xApp, using LLM meta-prompts to refine slice resource allocations for QoS iteratively. These works all show that structured prompting can yield effective RAN control policies. Most significantly, \citep{gajjar2025} introduced ORANSight-2.0, a comprehensive effort to develop foundational, open-source LLMs specifically for O-RAN. They presented RANSTRUCT, a novel framework for generating domain-specific instruction-tuning datasets. They demonstrated through extensive benchmarking that their fine-tuned models can outperform general-purpose and even closed-source alternatives on O-RAN tasks.

Standardization bodies have taken note. An O-RAN alliance report highlights how ``\textit{multi-domain network digital twins}'' can optimize RAN resource allocation, energy management, and fault detection \citep{orang2024}. 3GPP and TM Forum are formalizing intent frameworks for 5G RAN management \citep{clemm2022, bimo2025}. In parallel, agentic AI literature emphasizes the components needed for autonomous LLM agents (intentionality, planning, memory) \citep{ibm2025}.

\subsection{Appendix A: Detailed Formalism of the LLM-RAN Operator State, Action, and Reward Spaces}
We provide a more detailed mathematical definition of the core components of the LLM-RAN Operator framework introduced in Section~\ref{concept}. These formalisms are essential for grounding the theoretical analysis in the practical realities of a wireless network environment.

\subsubsection*{The State Space \(\mathcal{S}\)}

The state \(s_t \in \mathcal{S}\) at time \(t\) must capture a snapshot of the entire RAN environment. It is a high-dimensional, heterogeneous tuple. A comprehensive state representation can be defined as a product space:
\[
\mathcal{S} = \mathcal{H} \times \mathcal{Q} \times \mathcal{C} \times \mathcal{I}
\]
Where the component spaces are defined as follows:

\begin{itemize}
    \item \textbf{\(\mathcal{H}\) (Channel State Space):} This space represents the physical channel conditions for all \(K\) active users. It is a set of complex-valued matrices, \(\mathcal{H} = \{H_1, H_2, \dots, H_K\}\), where \(H_k \in \mathbb{C}^{N_R \times N_T}\) represents the channel between the \(N_T\) transmit antennas at the base station and the \(N_R\) receive antennas of user \(k\).

    \item \textbf{\(\mathcal{Q}\) (Queueing State Space):} This space captures the buffer status for all \(K\) users across \(F\) different traffic flows (e.g., eMBB, URLLC). It can be represented as a matrix \(Q \in \mathbb{R}_+^{K \times F}\), where each element \(Q_{k,f}\) is the current queue length (in bits or packets) for user \(k\)'s flow \(f\).

    \item \textbf{\(\mathcal{C}\) (Configuration Space):} This space describes the current configuration settings of the RAN. It is a structured tuple itself, e.g., \(\mathcal{C} = (P, B)\), where \(P = (p_1, \dots, p_M)\) is a vector of transmit power levels for the \(M\) cells, and \(B\) is a set of binary variables indicating which component carriers are currently active.

    \item \textbf{\(\mathcal{I}\) (Interference State Space):} This space characterizes the interference environment. It can be represented by a vector of Signal-to-Interference-plus-Noise Ratio (SINR) for each user, \(\mathbf{I} = (\text{sinr}_1, \dots, \text{sinr}_K)\), or by a more complex interference covariance matrix.
\end{itemize}

A full state vector \(s_t\) is therefore an instance \(s_t = (h_t, q_t, c_t, i_t)\) drawn from this product space.

\subsubsection*{The Action Space \(\mathcal{A}\)}

The action \(a_t \in \mathcal{A}\) is not a simple continuous vector but a structured, combinatorial command that modifies the network's configuration. The action space $\mathcal{A}$ consists of a set of valid configuration changes. An action can be formally defined as a function call with parameters:
\[
a_t = \{ \texttt{command\_type}, \texttt{parameters} \}
\]

Examples of specific commands include:

\begin{itemize}
    \item \textbf{Power Control:} \\
    \(a = \{\texttt{"set\_power"}, \{ (\texttt{cell\_1}, p'_1), (\texttt{cell\_2}, p'_2) \} \}\),\\
    subject to a constraint such as \(\sum p'_i \leq P_{\max}\).

    \item \textbf{Resource Block Allocation:} \\
    \(a = \{\texttt{"assign\_rbs"}, \{ (\texttt{user\_1}, [\texttt{rb}_1, \texttt{rb}_2]), (\texttt{user\_3}, [\texttt{rb}_5]) \} \}\).

    \item \textbf{Scheduler Weight Adjustment:} \\
    \(a = \{\texttt{"set\_scheduler\_weights"}, \{ (\texttt{flow\_URLLC}, w_1), (\texttt{flow\_eMBB}, w_2) \} \}\), \\
    where \(w_1 + w_2 = 1\).
\end{itemize}

The key characteristic is that \(\mathcal{A}\) is a discrete set of structured commands, and the task of the LLM-RAN operator is to generate a syntactically and semantically valid command from this set.

\subsubsection*{The Intent-Conditioned Reward Function \(R(s, a \mid i)\)}

The reward function \(R\) quantifies the desirability of a state transition and is critically conditioned on the operator's intent \(i\). Let the utility of a state \(s\) be a weighted sum of key performance indicators (KPIs), 
\[
U(s) = \mathbf{w}^T \cdot \texttt{kpi}(s)
\]
where \(\texttt{kpi}(s)\) is a vector of metrics such as \([\text{throughput}, \text{latency}, \text{energy\_consumption}]\). The reward for taking action \(a\) in state \(s\) is the change in utility:
\[
R(s, a) = U(s') - U(s), \quad \text{where } s' = f_{\text{env}}(s, a)
\]

The intent \(i\) directly influences the weight vector \(\mathbf{w}\). For example:

\begin{itemize}
    \item \textbf{Intent \(i_1\) = ``Maximize throughput'':} \\
    The LLM translates this to a weight vector \(\mathbf{w}_1 = [1, 0, 0]\), prioritizing the throughput KPI.

    \item \textbf{Intent \(i_2\) = ``Ensure ultra-low latency for hospital IoT devices'':} \\
    This translates to a more complex weight vector \(\mathbf{w}_2 = [\alpha, -\beta, \gamma]\), where \(\beta \gg \alpha, \gamma\), heavily penalizing latency.

    \item \textbf{Intent \(i_3\) = ``Reduce energy consumption during off-peak hours'':} \\
    This translates to a weight vector \(\mathbf{w}_3 = [0, 0, -1]\), negatively weighting energy consumption.
\end{itemize}

Thus, the LLM's role is not just to select an action \(a\), but to first interpret the intent \(i\) to define the very optimization problem (i.e., the reward function) that it is trying to solve.

\subsection{Appendix B: Expanded Argument for Proposition 3.1 (Expressiveness of LLM-Policies)}
\begin{prop}
    An LLM-RAN operator with sufficient model capacity can approximate any effective mapping from intents to RAN actions, assuming enough prompt context. In other words, the space of policies representable by a large transformer is universal for bounded RAN control tasks.
\end{prop}

\begin{proof}[Sketch]
    The proof follows from the established role of transformers as universal function approximators \citep{yun2020}. Any deterministic RAN control policy can be viewed as a mapping \(f: \mathcal{S} \rightarrow \mathcal{A}\). Given a sufficiently large model and a prompt context that effectively describes the state \(s\), the LLM can approximate this mapping \(f\). Empirical work showing LLMs matching or guiding bespoke RL controllers in networking tasks \citep{xingqi2025} provides evidence for this practical expressiveness. (A formal proof is beyond scope, but follows standard results on transformer expressivity.)
\end{proof}
We model the closed-loop LLM-RAN system as a discrete-time dynamical system: at step \(t\), the network is in state \(s_t \in \mathcal{S}\), the user issues intent \(i\), and the LLM operator outputs an action \(a_t = \mathcal{O}_{\rm LLM}(i,s_t)\). The network then transitions \(s_{t+1} = f_{\rm env}(s_t, a_t)\) according to physics and legacy control. This fits a Markov decision process (MDP) framework \citep{lingyan2025}. Indeed, in practice, each base station or RAN segment can be an MDP: the state space includes channel conditions, user demands, etc., plus any guidance from the LLM (e.g., an initial policy). The action space covers RAN parameters (power splits, scheduling weights \citep{lingyan2025}, and rewards reflect QoS or efficiency goals.

\subsubsection*{Argument}
The core of this proposition rests on two pillars: (1) the established universal approximation property of the Transformer architecture, and (2) the ability to represent the RAN control problem in a format amenable to these models.
\begin{itemize}
    \item[1.] \textbf{Universal Approximation Property of Transformers:} While early work on universal approximation focused on feedforward networks \citep{hornik1989}, recent research has extended these findings to the Transformer architecture that underpins modern LLMs. Specifically, \citep{yun2020} proved that a Transformer with sufficient capacity can approximate any continuous, permutation-equivariant sequence-to-sequence function. This provides the theoretical foundation that a sufficiently large Transformer can, in principle, model any well-behaved functional mapping between an input sequence and an output sequence.
    \item[2.] \textbf{Representing the RAN Control Problem as a Sequence-to-Sequence Task:} The primary challenge is to map our problem onto this sequence-to-sequence paradigm. A deterministic RAN control policy is a function \(f: \mathcal{S} \rightarrow \mathcal{A}\) which maps a state from the state space \(\mathcal{S}\) to an action in the action space \(\mathcal{A}\). To be approximated by an LLM, both the state and the action must be representable as sequences of tokens.

    Tokenizing the State Space (\(\mathcal{S}\)): The RAN state \(s \in \mathcal{S}\) is a high-dimensional, multi-modal vector containing diverse data types (e.g., channel matrices, user traffic queues, interference levels, hardware status). We can construct an effective input sequence for the LLM by tokenizing this state. This is a process of serialization and discretization. A plausible tokenization scheme would involve: 
    \begin{itemize}
        \item Using special tokens to demarcate different types of information (e.g., <CSI>, <QUEUES>, <CONFIG>).
        \item Flattening matrices (like channel state information) into one-dimensional vectors of quantized numerical values.
        \item Representing categorical data (e.g., modulation schemes) with their string names or integer codes.
        The result is a long, but finite, sequence of tokens that serves as the ``prompt context'' describing the current network state \(s_t\). For example: ``<STATE> <CSI> 0.91 0.23 \(\dots\) </CSI> <QUEUES> 1024 512 \(\dots\) </QUEUES> \(\dots\)''
    \end{itemize}
    Generating the Action Sequence (\(\mathcal{A}\)):  Similarly, the action \(a \in \mathcal{A}\) can be represented as a structured sequence. Since actions are often configurations, a natural format is a sequence of key-value pairs, similar to JSON or a domain-specific language (DSL). For example, the desired output sequence for an action could be ``<ACTION> \(\text{set}_{\text{power}}\)(\(\text{cell}_1\), -10dBm); \(\text{set}_\text{scheduler}\)(weights=[0.8, 0.2]) </ACTION>''
    \item[3.] \textbf{Connecting the Pillars:} With these representations, the RAN control policy \(f: \mathcal{S} \rightarrow \mathcal{A}\) becomes a mapping from an input token sequence (representing \(s_t\)) to an output token sequence (representing \(a_t\)). This is precisely the class of problems that Transformers are proven to be able to approximate \citep{yun2020}. Therefore, an LLM-RAN operator, which is a Transformer-based model, can approximate any such policy function \(f\), provided the model has sufficient capacity and is trained or fine-tuned on relevant data.
\end{itemize}
Empirical results from related work, such as \citep{xingqi2025} and \citep{gajjar2025}, provide strong evidence for this proposition in practice. They show that LLMs can indeed learn to generate practical control actions or configurations for networking tasks, implicitly demonstrating that they are successfully approximating complex policy functions.

\subsection{Appendix C: Formal Proof of Lemma 3.2 (Monotonic Improvement)}
\begin{lemma}
    Let \(U(s)\) be a utility function representing a network performance metric. Assume the LLM operator is designed to solve the single-step optimization problem \(a_t = \arg \max_{a \in \mathcal{A}}U(f_{\text{env}}(s_t, a))\). If a solution \(a_t\) exists such that \(U(s_{t + 1}) \geq U(s_t)\), then the sequence of states generated by the system is monotonically improving in utility. This property is observed in LLM-guided cases (e.g., decreasing transmit power raises efficiency).
\end{lemma}

\begin{proof}
   \begin{enumerate}
   \item We seek to prove that for any time step t, the utility of the next state, \(s_{t+1}\), is greater than or equal to the utility of the current state, \(s_t\).
    By the definition of the system's dynamics, the state at time \(t + 1\) is given by the application of the environment function to the current state \(s_t\) and the chosen action \(a_t\):
    \begin{equation*}
        s_{t+1} = f_{\text{env}}(s_t, a_t)
    \end{equation*}
    \item The action space \(\mathcal{A}\) must contain, either explicitly or implicitly, a ``do-nothing'' or identity action, which we will denote as \(a_{\text{null}}\). This action is defined such that it does not change the state of the network. Therefore, applying the environment dynamics with this action yields the same state:
    \begin{equation*}
        f_{\text{env}}(s_t, a_{\text{null}}) = s_t
    \end{equation*}
    This action \(a_{\text{null}}\) is a member of the set of all possible actions, \(\mathcal{A}\).
    \item According to the central assumption of the lemma, the action \(a_t\) is chosen to be the optimal action that maximizes the utility \(U\) of the resulting state. This means that \(a_t\) must yield a utility that is greater than or equal to the utility produced by any other possible action \(a' \in \mathcal{A}\). 
    \begin{equation*}
        U\left(f_{\text{env}}(s_t, a_t)\right) \geq U\left(f_{\text{env}}(s_t, a')\right) \quad \forall a' \in \mathcal{A}
    \end{equation*}
    \item Since \(a_{\text{null}}\) is a member of \(\mathcal{A}\), the above inequality must also hold for \(a' = a_{\text{null}}\):
    \begin{equation*}
        U(f_{\text{env}}(s_t, a_t)) \geq U(f_{\text{env}}(s_t, a_{\text{null}}))
    \end{equation*}
    \item By substituting the definitions from steps 1 and 2 into the inequality from step 4, we arrive at:
    \begin{equation*}
        U(s_{t+1}) \geq U(s_t)
    \end{equation*}
    \end{enumerate}
    Since this holds for any arbitrary time step \(t\), the sequence of utilities (\(U(s_t)\)) is monotonically non-decreasing. \textbf{QED}
\end{proof}

\subsection{Appendix D: Extended Discussion on Theorem 3.3 (Convergence Conditions)}
\begin{theorem}
    If the closed-loop operator \(F(s) \triangleq f_{\text{env}}\left(s, O_{\text{LLM}}(i, s)\right)\) forms a contraction mapping on the state space \(\mathcal{S}\) concerning a suitable norm, then the sequence of states (\(s_t\)) generated by the system converges to a unique fixed point \(s^*\).
\end{theorem}
\subsubsection{Discussion}
The proof of this theorem is a direct application of the Banach Fixed-Point Theorem. The theorem's power, however, lies not in its direct application to the general problem\textemdash which is likely intractable\textemdash but in its ability to provide a formal language for dissecting the immense challenges of guaranteeing stable AI-driven network control. The critical assumption is that the composite function \(F(s)\) is a contraction mapping. A function \(F\) is a contraction if there exists a constant \(k\), where \(0 \leq k < 1\), such that for any two states \(s_1, s_2 \in S\), the following inequality holds:
\[
\|F(s_1) - F(s_2)\| \leq k \cdot \|s_1 - s_2\|
\]
This condition, while simple to state, presents several profound research hurdles when applied to the LLM-RAN operator context. We discuss the three primary challenges below.

\textbf{The Challenge of Defining a Suitable Norm on \(\mathcal{S}\)}

The state space \(\mathcal{S}\) of a RAN is not a simple Euclidean space. As outlined in Appendix A, it is a complex, multi-modal product space containing:
\begin{itemize}
    \item Continuous variables (e.g., floating-point values in channel matrices).
    \item Discrete integer variables (e.g., traffic queue lengths).
    \item Categorical variables (e.g., modulation schemes, service classes).
\end{itemize}
Defining a meaningful distance metric or norm \(\|\cdot\|\) on such a heterogeneous space is a significant challenge in itself. A standard \(L_2\)-norm may not be appropriate, as it would treat a change in a channel coefficient as equivalent to a change in a queue length, which may not be operationally meaningful. The choice of norm fundamentally affects whether a function can be proven to be a contraction, and identifying a norm that captures the ``operational distance'' between two RAN states is a significant open problem.

\textbf{The Nature of the Wireless Environment \(f_{\text{env}}\)}

The environment dynamics \(f_{\text{env}}\) are governed by the physics of wireless communication, which are inherently hostile to the assumptions of contraction mappings.
\begin{itemize}
    \item \textbf{Stochasticity:} Wireless channels are stochastic due to multipath fading, shadowing, noise, and user mobility. Consequently, $f_{\text{env}}$ is not a deterministic function, meaning the Banach theorem cannot be directly applied. One would need to resort to more complex stochastic fixed-point theorems, which come with their own sets of stringent requirements.
    \item \textbf{Non-Linearity:} The relationship between actions (e.g., power allocation) and outcomes (e.g., SINR) is highly non-linear and non-convex, making it very unlikely that \(f_{\text{env}}\) is globally contractive.
    \item \textbf{Time-Variation:} The environment is non-stationary. As users move and conditions change, the function \(f_{\text{env}}\) itself changes over time (\(f_{\text{env},t}\)), violating the assumption of a single, fixed function \(F\).
\end{itemize}

\textbf{The Complexity of Bounding the LLM Operator \(O_{\text{LLM}}\)}

Even if the environment were simple and deterministic, bounding the behavior of the LLM operator \(O_{\text{LLM}}\) is a frontier research problem in deep learning.
\begin{itemize}
    \item \textbf{Lipschitz Constant of NNs:} Proving that a deep neural network is Lipschitz continuous, let alone calculating its Lipschitz constant \(k\), is notoriously difficult. For models with billions of parameters, complex attention mechanisms, and non-linear activations, obtaining tight bounds is often computationally intractable.
    \item \textbf{Discontinuities from Tokenization:} The LLM operates on discrete tokens. The process of tokenizing a continuous state \(s\) can introduce sharp discontinuities. An infinitesimally small change in a continuous value in \(s\) could, in principle, alter the tokenized sequence, leading to a drastically different output from \(O_{\text{LLM}}\). This makes the operator potentially non-continuous, violating a prerequisite for being a contraction.
\end{itemize}

\subsubsection*{A Path Forward: The Framework as an Analytical Tool}

Given these hurdles, it is clear that proving global convergence for a general LLM-RAN operator is not a realistic short-term goal. However, the value of Theorem 3.3 is that it provides a formal roadmap for research. It allows the community to investigate convergence under simplified, tractable conditions. For instance, researchers can now ask more precise questions, such as:
\begin{itemize}
    \item ``Under a simplified, deterministic channel model, can we design \(O_{\text{LLM}}\) to be a contraction?''
    \item ``If we restrict the action space \(A\) to a finite set, can we prove convergence to a stable cycle?''
    \item ``Can we design the Control Plane Adapter to explicitly enforce a Lipschitz constraint on the final output, thereby guaranteeing stability?''
\end{itemize}

Thus, our framework transforms a vague goal of ``making AI stable'' into a set of concrete mathematical problems that can be tackled incrementally.

\end{document}